\documentclass[fleqn,preprint,12pt,3p]{elsarticle}

\usepackage{amsmath}
\usepackage{amssymb}
\usepackage{amsthm,subcaption,caption}

\usepackage{algorithm}
\usepackage{algorithmic}
\usepackage{pseudocode}
\usepackage{multirow}
\usepackage{booktabs}
\usepackage{subfig}
\usepackage{amsmath,amssymb,amsthm}
\usepackage{epstopdf}

\usepackage{graphicx}
\usepackage{epsfig}
\newtheorem{theorem}{Theorem}[section]
\newtheorem{lemma}[theorem]{Lemma}

\newtheorem{proposition}[theorem]{Proposition}

\newdefinition{definition}[theorem]{Definition}
\newdefinition{remark}[theorem]{Remark}
\newdefinition{example}[theorem]{Example}

\begin{document}

\begin{frontmatter}



\title{Deterministic construction of sparse binary and ternary matrices from existing binary sensing matrices}


%\author{Pradip Sasmal\footnote{author for correspondence Telephone: 091-40-23016072, Fax: 091-40-23016032 } R. Ramu Naidu, and C. S. Sastry \\ Department of Mathematics \\ Indian Institute of Technology, Hyderabad-502205, India \\
%Email: \{ma12p005, ma11p003, csastry\} @ iith.ac.in \\ P. V. Jampana \\ Department of Chemical Engineering\\ Indian Institute of Technology, Hyderabad-502205, India \\
%Email: pjampana@ iith.ac.in}

\author{Pradip Sasmal\footnote{Author for correspondence Telephone: 091-40-23016072, Fax: 091-40-23016032 }, R. Ramu Naidu and C. S. Sastry}
\address{Department of Mathematics \\ Indian Institute of Technology, Hyderabad-502205, India \\
Email: \{ma12p1005, ma11p003, csastry\} @ iith.ac.in}

\author{P. V. Jampana}
\address{Department of Chemical Engineering\\ Indian Institute of Technology, Hyderabad-502205, India \\
Email: pjampana@ iith.ac.in}
\begin{abstract}
In the present work, we discuss a procedure for constructing sparse binary and ternary matrices from existing two binary sensing matrices. The matrices that we construct have several attractive properties such as smaller density, which supports algorithms with low computational complexity. As an application of our method, we show that a CS matrix of general row size different from $p, p^2, pq$ (for different primes $p,q$) can be constructed. 
%In Compressed Sensing the matrices that satisfy the Restricted Isometry Property (RIP) play an important role.
% But  to date, very few results  for designing such matrices are available. Of interest in several applications is a
 %binary sensing matrix. The present work constructs the deterministic and binary sensing matrices using the Euler Squares. In particular, given a positive integer $m$ different from $p, p^2$ for a prime $p$, we show that it is possible to construct a binary sensing matrix  whose row size is $m$ and the column size is a function of $m^2$.
\end{abstract}


\end{frontmatter}

\section{Introduction}

Recent developments at the intersection of algebra and optimization theory, by the name of Compressed Sensing (CS), aim at providing sparse descriptions to linear systems.  These developments are found to have tremendous potential for several applications \cite{bourgain_2011, CS_2013, gil_2010, Chand_2014, Lura_2013}. Sparse representations of a function are a powerful analytic tool in many application areas such as image/signal processing and numerical computation \cite{Bruckstein_2009}, to name a few. The need for the sparse representation arises from the fact that several real life applications demand the representation of data in terms of as few basis (frame) type elements as possible. The elements or the columns of the associated matrix $\Phi$ are called atoms and the matrix so generated by them is called the dictionary. The developments of CS Theory depend typically on sparsity and incoherence \cite{Bruckstein_2009}\cite{Kashin_2007}. Sparsity expresses the idea that the “information rate” of a continuous time data may be much smaller than suggested by its bandwidth or that a discrete-time data depends on a number of degrees of freedom which is comparably much smaller than its (finite) length. On the other hand, incoherence extends the duality between the time and frequency contents of data. 
\par Since the sparsity of the CS matrix is key to minimizing the computational complexity associated with the matrix-vector multiplication, it is desirable that a CS matrix has smaller density. Here, by density, one refers to the ratio of number of nonzero entries to the total number of entries of the matrix. The sparse CS matrix may contribute to fast processing with low computational complexity in Compressed Sensing \cite{gil_2010}.

\par In the recent literature on CS \cite{adcock_2013, bryant_2014, dima_2012, mixon_2012, ind_2008}, deterministic construction of CS matrices has gained momentum.  R. Devore \cite{Ronald_2007} has constructed deterministic binary sensing  matrix of size $p^2 \times p^{r+1}$, where $p$ is prime or prime power. The density of this matrix is $\frac{1}{p}.$ S. Li et. al. \cite{li_2012} have generalized Devore's work, constructing binary sensing matrix of size $|\mathcal{P}|q \times q^{\mathcal{L}(G)}$, where $q$ is any prime power and  $\mathcal{P}$ is the set of all rational points on algebraic curve $\mathcal{X}$ over finite field $\mathbb{F}_q$. The density of this matrix is $\frac{1}{q}$. % P. Indyk \cite{ind_2008} has constructed binary sensing matrices using hash functions and extractor graphs with sizes $r2^{O(\log\log n)^{O(1)}} \times n$, where $r \ll n$.
 A. Amini et. al. \cite{amini_2011} have constructed binary sensing matrices %with sizes $(16^{a}-1) \times \frac{(16^{a}-1)(16^{a}-6)}{5}$ 
using OOC codes. The density of this matrix is $\frac{\lambda}{m}$, where $m$ is row size and $\lambda$ is the number of ones in each column. In all these constructions row sizes of associated matrices are given by some particular family of numbers.

\par CS matrices of general size will not impose any restriction on the data to be dealt with, when used in  applications such as data compression or classification.
The main contribution of the present work may be summarized as follows:
\begin{itemize}
\item Construction of sparse binary and ternary matrices from existing binary matrices.
\item As an application of our construction methodology, we show that, given $m \neq p,p^{2},pq$ (for different primes $p,q$), it is possible to construct sparse binary and ternary matrices of row size $m$.
\end{itemize}
\noindent Consequently, we believe that this work is an attempt towards constructing CS matrices of general size.

%\par In the present work, however, we construct binary and ternary CS matrices of very high sparsity (that is, low density) from existing binary sensing matrices. As an application of proposed method, we show that a CS matrix of general size different from $p, p^2, pq$ (for different primes $p,q$) can be constructed.
\par The paper is organized in several sections. In section 2, we provide basics of CS theory.  While in sections 3 and 4, we discuss in detail the construction of CS matrices of low density and an application of proposed method respectively. In the last two sections, we present respectively the construction of ternary matrices and concluding remarks.

\section{Basics of Compressed Sensing:}

\par The objective of compressed sensing is to recover $x = (x_{1},x_{2}\ldots,x_{M})^{T}\in \mathbb{R}^{M} $ from a few of its linear measurements $y\in \mathbb{R}^{m} $ through a stable and efficient reconstruction process via the concept of sparsity. From the measurement vector $y$ and the sensing mechanism, one  gets a system $y=\Phi x$, where $\Phi$ is an $m \times M \;(m < M)$ measurement matrix. 

Denoting $\Phi^{r}_{k}$ as the $k^{\mbox{th}}$ row of $\Phi$, one may rewite the $k^{\mbox{th}}$ component in $y$ as $y_{k}=\left\langle \Phi^{r}_{k} ,x \right\rangle, k=1,2,\dots m$. Here $\left\langle \Phi^{r}_{k},x\right\rangle$ represents the inner-product between $\Phi^{r}_{k}$ and $x$. That is, the object $x$ to be acquired is correlated with the waveform $\Phi^{r}_{k}$. This is a standard setup in several applications \cite{Donoho_2006}\cite{candes_2008}. For example, if the sensing waveforms are Dirac delta functions, then $y$ is a vector of sampled values of $x$ in time or space domain. If the sensing waveforms are indicator functions of pixels, then $y$ is the image data typically collected by sensors in a digital camera. If the sensing waveforms are sinusoids, then $y$ is a vector of Fourier coefficients and this modality is used in the magnetic resonance imaging. Nevertheless, if the sensing waveforms have $0$ and $1$ (or $0$ and $\pm1$) as elements, then the associated  matrix (referred conventionally to as a sensing matrix) can have potential application for multiplier-less dimensionality reduction. An excellent overview of Compressed Sensing and the applicability of several sensing matrices may be seen in \cite{candes_2008}.

Given the pair $(y, \Phi)$, the problem of recovering $x$ can be formulated as finding the sparsest solution (that is, the solution containing most number of zero entries) of linear system of equations $y=\Phi x$. Sparsity is measured by $\| . \|_{0}$ norm and $\|x\|_{0} = |\{j: x_j \neq 0\}|$, the number of non-zero entries in $x$. 
%\footnote{it violates the triangular inequality, that is $\| x_{1} + x_{2} \|_{0} \geq \| x_{1}\|_{0} + \|x_{2} \|_{0} $.}).
 Now finding the sparsest solution  can be formulated as the following minimization problem (generally denoted as $P_{0}$ problem):
%\begin{displaymath}
\begin{equation}
P_0: \min_{x}\|x\|_{0}  \quad \mbox{subject to} \quad \Phi x=y.
%\lebel{}
\end{equation}

%\end{displaymath} 

This $P_{0}$ problem is a combinatorial minimization problem and is known to be NP-hard \cite{bourgain_2011}. One may use greedy methods and convex relaxation of $P_{0}$ problem to recover $k-$sparse signals (that is, $\|x\|_{0}=k$ or signals with $k$ number of nonzero entries). The convex relaxation of $P_{0}$ problem can be posed as $P_{1}$ problem \cite{can_2008, Candes_2005}, which is defined as follows:
\begin{equation}
P_1:\min_{x} \|x\|_{1} \quad \mbox{subject to} \quad \Phi x=y.
%\lebel{}
\end{equation}
\par The Orthogonal Matching Pursuit (OMP) algorithm and and the $l_{1}-$norm minimization (also called basis pursuit) are two widely studied
CS reconstruction algorithms \cite{troop_2010}. 

\par Candes and Tao \cite{Tao_2006} have introduced the following isometry condition on matrices $\Phi$ and have established its important role in CS. An $m \times M$ matrix $\Phi$ is said to satisfy the Restricted Isometry Property(RIP) of order $k$ with constant $\delta_{k}$ $( 0<\delta_{k}<1)$ if for all vectors $x\in \mathbb{R}^{M}$ with $\|x\|_{0}\leq k$, we have
\begin{equation} 
\label{eqn:rip}
(1-\delta_{k}) \left\|x\right\|^{2}_{2} \leq \left\|\Phi x\right\|^{2}_{2} \leq (1+\delta_{k}) \left\|x\right\|^{2}_{2}.
\end{equation} 
%Equivalently, for all vectors $x\in \mathbb{R}^{M}$ with $\left\|x\right\|_{2} = 1$ and $\|x\|_{0}\leq k$, one may rewrite (\ref{eqn:rip}) as
%\begin{displaymath}
%\label{eqn:r}
%(1-\delta_{k})  \leq \left\|\Phi x\right\|^{2}_{2} \leq (1+\delta_{k}) .
%\end{displaymath}
The following theorem \cite{Candes_2005} establishes the equivalence between $P_{0}$ and $P_{1}$ problems: 
\begin{theorem}
Suppose an $m \times M$ matrix $\Phi$ has the $(2k, \delta)$ restricted isometry property for some $\delta < \sqrt{2}-1$, then $P_{0}$ and $P_{1}$ have same $k-$sparse solution if $P_{0}$ has a $k-$sparse solution.
\end{theorem}
\par The mutual-coherence $\mu(\Phi)$ of a given matrix $\Phi$ is the largest absolute inner-product between different  normalized columns of $\Phi$, that is, $\mu(\Phi)= \max_{1\leq\; i,j \leq\; M,\; i\neq j} \frac{|\; \Phi_i ^T\Phi_j|}{\Vert \Phi_i\Vert_{2} \Vert \Phi_j \Vert_2}$. Here, $\Phi_k$ stands for the $k$-th column in $\Phi$. 
\noindent The following proposition \cite{bourgain_2011} relates the RIP constant $\delta_{k}$ and $\mu.$
\begin{proposition}
\label{thm:pro}
Suppose that $\Phi_{1},\ldots,\Phi_{M}$ are the unit norm columns of the matrix $\Phi$ with coherence $\mu$. Then $\Phi$ satisfies RIP of order $k$ with constant $\delta_{k} = (k-1)\mu$.
\end{proposition}

\par One of the important problems in CS theory deals with constructing CS matrices that satisfy the RIP for the largest possible range of $k$. It is known that the widest possible range of $k$ is of the order $\frac{m}{\log(\frac{M}{m})}$, for some well-behaved order constant \cite{Richard_2008, Ronald_2007, Kashin_1978}. However the only known matrices that satisfy the RIP for this range are based on random constructions \cite{Richard_2008, Candes_2005}. Presently the researchers working in CS theory attempt issues \cite{adcock_2013, apple_2009, mixon_2014} such as 
\begin{enumerate}
\item Maximizing the sparsity of the solution, that is, for a given pair $(y,\Phi )$, finding $x$ such that  $\|x\|_0$  is as small possible
\item	Improving the recovery process, that is, incorporating the prior information (for example, partial support solution) into the recovery process
\item Developing faster algorithms
\item	Designing matrices that satisfy recovery properties
\end{enumerate}
\par The recovery algorithms available to date use matrices $\Phi$  that are generated at random and there is no efficient method for verifying whether or not a given matrix $\Phi$ does indeed have the stronger reconstruction properties. For this and other reasons, it is useful to have an explicit construction of $\Phi$.

\section{ Construction Procedure:}
%%%%%%%%%%%%%%%%%%%%%%%%%%%%%%%%%%%%%%%%%%%%%%%%%%%%%%%%%%%%%%%%%%%%%%%%%%%%%
\par In this section, we present our deterministic construction procedure taking two existing binary sensing matrices as inputs. We show that the resulting matrix has smaller density than the original two. Let $\Psi_{m \times M}$ be a binary sensing block matrix having $k''$ row blocks (the meaning of row blocks is explained later with example), each of size $n$ such that every block contains single $1$ and the overlap between any two columns is atmost $r$, that is $|\{l|\Psi_{l,p}=1\} \bigcap \{t|\Psi_{t,q}=1 \}|\leq r$ for any two columns $\Psi_{p}$ and $\Psi_{q}$. Let $\Psi'_{m' \times M'}$ be another binary sensing block matrix having $k'$ row blocks, each of size $n'$ such that each block contains single $1$ and the overlap between any two columns is atmost $r'$. Assume that $n' \leq n$ , $ k \leq \min \{ k' , k'' \}$ and $r \leq r' \leq k\leq n$. 
%Define the vector $V$, indexed on $i-1$(mod $n$), for $1 \leq i \leq nk=m$. There are $k$-blocks in $V$ and each block contains the numbers $0,1,\ldots,n-1$.
 Now a new matrix can be constructed by the following steps:

{\bf Step-1:} Let the $i^{th}$ column of $\Psi$ be  $\Psi_{i}$. For $1 \leq i \leq M$, define $S_{i} =  (((\text{supp}(\Psi_{i})-1)(\mod \;n))^{T}+ \mathbf{1} ^{T})$, where $\mathbf{1} ^{T}$ is the vector consisting of all ones of size $k'' \times 1$, supp$(\Psi_{i})$ is the support vector of $\Psi_{i}$. Since $\Psi$ has $k''$ blocks and each block contains one $1$, it follows that $|S_{i}| = k''$ where $S_{i}$ is a $k''-$ tuple on the set $X=\{1,2,\ldots,n\}$. Since $\Psi$ has $M$ columns, we have $M$ such $k''-$tuples. For example, suppose $m=9$ and $\Psi$ has 3 blocks, then each block is of size 3. Now, if the $i^{th}$ column $\Psi_{i}$ is $(1\; 0\; 0\; 0\; 1\; 0\; 0\; 0\; 1)^{T}$, supp$(\Psi_{i})$ is $[1\;5\;9]$, then the triplet, $S_{i}$, corresponding to this column  is $[1\;2\;3]$. Similarly from the matrix $\Psi'$, we can generate $M'$ number of $k'-$tuples on the set $Y = \{1,2,\ldots,n'\}.$  
%%%%%%%%%%%%%%%%%%%%%%%%%%%%%%%%%%%%%%%%%%%%%%%%%%%%%%%%%%%%%%%%%%%%%%%%%%

{\bf Step-2:} From $k''-$tuples of first matrix $\Psi$ we remove last $(k'' - k)$ entries from each tuple to make it a $k-$tuple.  Now add $(-1)$ to each of the entries of the $M$ number of $k-$tuples that are obtained from $\Psi$. Let the $k-$tuples be of the form $(c_{i1}, c_{i2}, \ldots,c_{ik})$ for $1 \leq i \leq M$. From $k'-$tuples of second matrix $\Psi'$, we remove last $(k' - k)$ entries from each tuple to make it $k-$tuple. By this process, we get $M'$ number of $k-$tuples from the second matrix $\Psi'$. Let each $k-$tuple have the form $(c'_{j1}, c'_{j2}, \ldots,c'_{jk})$ for $1 \leq j \leq M'$.

{\bf Step-3:} Now replace each $k-$tuple $(c_{i1}, c_{i2}, \ldots,c_{ik})$ with $M'$ number of $k-$tuples (obtained from $\Psi'$) by adding $n' (c_{i1}, c_{i2}, \ldots,c_{ik})$  to each of the $k-$tuples $(c'_{j1}, c'_{j2}, \ldots,c'_{jk})$ for $1 \leq j \leq M'$. This way, we get $MM'$ number of $k-$ tuples on the set $X'=\{1,2,\ldots,nn'\}$. Denote these $k-$tuples by $F$. Let $(a_{j1}, a_{j2}, \ldots,a_{jk})$ for $j= 1,2,....,MM'$ be the $k-$ tuples in $F$.

{\bf Step-4:} From these $k-$ tuples we form a binary vector of length $knn'$ where $1$ occurs in the positions $(l-1).nn'+a_{jl}$ for $l= 1,2,\dots,k$ and rest of the positions are zeros. Using these $MM'$ number of $k-$tuples, we form binary sensing matrix $\Phi$ having $k$ number of blocks and each block is of size $nn'$ and every block contains single $1$. The position of one's in each block is indexed by the $k-$ tuples. So the size of the matrix $\Phi$ becomes $nn'k \times MM'$.
The pseudo code of the proposed method is as follows: \\
%----------------------------------------------
\begin{center}
\begin{tabular}{|l l|}
\hline \\
1. & {\bf Input:}  Two matrices $\Psi$, $\Psi'$ of sizes $nk'' \times M$, $n'k' \times M'$ \\
2. & Suppose $|\text{supp}(\Psi_i) \cap \text{supp}(\Psi_j)| \leq r $ and \\
   & $ |\text{supp}(\Psi'_i) \cap \text{supp}(\Psi'_j)| \leq r'$, $\forall \ i \neq j$  \\
3. & Assume  $n' > n, k \leq \min \{k',k''\}$ and $r \leq r' \leq k \leq n$ \\
4. & Set $S_{i} =  (((\text{supp}(\Psi_{i})-1)(\mod \;n))^{T}+ \mathbf{1} ^{T})$, $|S_i| = k'', \; \forall \; i=1,2,\ldots,M$  \\
5. & Similarly, $S'_j$ with $|S'_j| = k'$ are defined for $j=1,2,\ldots,M'$ \\
6. & Set $S''_{i,j} = S'_{j,k} - \mathbf{1}+n' S_{i,k}$, $S'_{j,k}$ and $S_{i,k}$ are first $k$ entries of $S'_j$ and $S_{i}$ \\
7. & $|S''_{i,j}| = k$ and Let $S''_{i,j} = (S''_{i,j,1}, S''_{i,j,2}, \dots, S''_{i,j,k}) $ \\
   &  $1 \leq S''_{i,j,l} \leq nn', \forall l=1,2,...,k$ \\
8. & From each $S''_{i,j}$, create a vector $v_{i,j}$ of length $nn'k$ \\ 
   &  and $v_{i,j}=1$ at $(l-1).nn'+ S''_{i,j,l}$ for $l= 1,2,\dots,k$ and zero elsewheres\\  
9. & {\bf Output:} $\Phi'$, a matrix of size $nn'k \times MM'$, whose columns are $v_{i,j}.$ \\ \\
\hline
\end{tabular}
\end{center}
%%%%%%%%%%%%%%%%%%%%%%%%%%%%%%%%%%%%%%%%%%%%%%%%%%%%%%%%%%%%%%%%%%%%%%%%%%%%%%%%%%%%%

%This methodology is summarized in Fig.1
\begin{lemma}\label{lem:1}
The coherence of the enlarged matrix $\Phi$ of size $nn'k \times MM'$, $\mu(\Phi)$, is atmost $\frac{r'}{k}$.
\end{lemma}

\begin{proof}
Let $\Phi_{i}, \Phi_{j}$ be two arbitrary columns of matrix $\Phi$. There exist two $k-$ tuples $f_i ,f_j \in F$ such that $\Phi_{i}, \Phi_{j}$ are the corresponding vectors of $k-$ tuples $f_i ,f_j$ as defined in the above construction . Suppose $f_{i}= (c'_{k'1}, c'_{k'2}, \ldots,c'_{k'k})+n'(c_{k1}, c_{k2}, \ldots,c_{kk})$ and $f_{j}= (d'_{k'1}, d'_{k'2}, \ldots,d'_{k'k})+n'(d_{k1}, d_{k2}, \ldots,d_{kk})$, where $(c'_{k'1}, c'_{k'2}, \ldots,c'_{k'k}), (d'_{k'1}, d'_{k'2}, \ldots,d'_{k'k})$ are two $k-$tuples obtained from $\Psi'$ and $(c_{k1}, c_{k2}, \ldots,c_{kk}), (d_{k1}, d_{k2}, \ldots,d_{kk})$ are $k-$tuples obtained from $\Psi$.
Set $g_{ij} = f_{i}-f_{j}= (c'_{k'1}, c'_{k'2}, \ldots,c'_{k'k})-(d'_{k'1}, d'_{k'2}, \ldots,d'_{k'k}) + n'\biggl((c_{k1}, c_{k2}, \ldots,c_{kk})-(d_{k1}, d_{k2}, \ldots,d_{kk})\biggl).$ We will show that $g_{ij}$ has atmost $r'$ number of zero elements, which implies that intersection between the supports of $\Phi_{i}, \Phi_{j}$ is atmost $r'$. Now $|c'_{k^{'}i}-d'_{k^{'}i}|<n'$ and $|c_{ki}-d_{ki}|<n$. The $l^{th}$ element of $g_{ij}$ is $(g_{ij})_{l}= (c'_{k'l}-d'_{k'l})+ n'(c_{kl}-d_{kl}).$ We investigate the cases wherein $(g_{ij})_{l} =0.$\\
\textbf{Case 1:} Suppose $c'_{k'l} \neq d'_{k'l}$. \\
\textbf{Sub case 1.1:} If $c_{kl} = d_{kl}$, then $(g_{ij})_{l} \neq 0.$\\
\textbf{Sub case 1.2:} If $c_{kl} \neq d_{kl}$, then also $(g_{ij})_{l} \neq 0.$ Since $1 \leq |c_{kl}-d_{kl}|<n$ and $|c'_{k^{'}l}-d'_{k^{'}l}|<n'$, we have $(g_{ij})_{l} = (c'_{k^{'}l}-d'_{k^{'}l})+ n'(c_{kl}-d_{kl}) \neq 0$.\\
\textbf{Case 2:} Suppose $c'_{k'l} = d'_{k'l}$ then $(g_{ij})_{l} =0$ only when $c_{kl} = d_{kl}$.
\par From the above cases, we conclude that $(g_{ij})_{l} =0$ only when $c'_{k'l} = d'_{k'l}$ and $c_{kl} = d_{kl}$, but as  both matrices $\Psi$ and $\Psi'$ can have atmost $r'$ intersections, $(g_{ij})_{l} =0$ can occur for atmost $r'-$choices of $j$. So $\Phi_{i}, \Phi_{j}$ have atmost $r'-$intersections. So the coherence $\mu(\Phi)$ of $\Phi$ is atmost $\frac{r'}{k}$.

\end{proof} 
\noindent The following theorem shows the RIP compliance of $\Phi$. 
\begin{theorem}
\label{thm:r}
The afore-constructed matrix $\Phi_{nn'k \times MM'}$ satisfies RIP with $\delta_{k'} = (k'-1)({\frac{r'}{k}})$ for any $k' < \frac{k}{r'} + 1.$    
\end{theorem}
\begin{proof}
Proof follows from the Proposition \ref{thm:pro} and Lemma \ref{lem:1}. 
\end{proof}
\noindent \textbf{Remark-1:} The density of the matrix $\Phi$ is $\frac{1}{nn'}$, which is smaller than the densities of the matrices $\Psi$ and $\Psi'$, which are $\frac{1}{n}$ and $\frac{1}{n'}$ respectively. Hence using our construction procedure we can construct more sparse binary matrices from existing binary sensing matrices.\\
\noindent \textbf{Remark-2:} The positions of $1'$s in the matrix $\Phi$ are known and we know that every block contains single one. So we can permute the columns of the matrix $\Phi$ so that the support of the permuted matrix has nice structured sparsity.  
\subsection{Example:} In this section, we demonstrate our proposed method of construction via an example.
\noindent Let $\Psi_{4 \times 4}$ be a binary sensing block matrix

%$\vspace{10pt}$$\begin{pmatrix}
%1 & 0 & 1 & 0 \\
%0 & 1 & 0 & 1 \\
%\hline
%1 & 0 & 0 & 1 \\
%0 & 1 & 1 & 0 \
%
%\end{pmatrix}$

\[ \left( \begin{array}{cccc}                       
     1 & 0 & 1 & 0 \\
     0 & 1 & 0 & 1 \\
		\hline
     1 & 0 & 0 & 1 \\
     0 & 1 & 1 & 0 \end{array} \right)\] 
having $2$ row blocks, each of atmost size $2$ such that every block contains one $1$ and the overlap between any two columns is $1$. Let $\Psi'_{9 \times 9}$ be a binary sensing block matrix 
\[ \left( \begin{array}{ccccccccc}

     1 & 0 & 0 & 1 & 0 & 0 & 1 & 0 & 0 \\
     0 & 1 & 0 & 0 & 1 & 0 & 0 & 1 & 0 \\
     0 & 0 & 1 & 0 & 0 & 1 & 0 & 0 & 1 \\
		\hline
     1 & 0 & 0 & 0 & 0 & 1 & 0 & 1 & 0 \\
     0 & 1 & 0 & 1 & 0 & 0 & 0 & 0 & 1  \\
     0 & 0 & 1 & 0 & 1 & 0 & 1 & 0 & 0  \\
		\hline
     1 & 0 & 0 & 0 & 1 & 0 & 0 & 0 & 1  \\
     0 & 1 & 0 & 0 & 0 & 1 & 1 & 0 & 0  \\
     0 & 0 & 1 & 1 & 0 & 0 & 0 & 1 & 0 \end{array} \right)\]
having $3$ row blocks, each of atmost size $3$ such that every block contains one $1$ and the overlap between any two columns is $1$.

{\bf Step-1:} The corresponding four $2-$tuples obtained from $\Psi$ constitute $S$ as
 
$S=\Bigg\{ \begin{pmatrix}
1\\
1
\end{pmatrix},\begin{pmatrix}
2\\
2
\end{pmatrix},\begin{pmatrix}
1\\
2
\end{pmatrix},\begin{pmatrix}
2\\
1
\end{pmatrix} \Bigg\}$ $\mathbf{.}$

%
%\[ \left( \begin{array}{cccc}
      %
			%2 & 1 & 2 & 1 \\
			%2 & 1 & 1 & 2 \end{array} \right)\]
			
 The corresponding nine $3-$tuples obtained from $\Psi'$ constitute $S'$ as

$S'=\Bigg\{\begin{pmatrix}
1\\
1\\
1\end{pmatrix},
\begin{pmatrix}
2\\
2\\
2\end{pmatrix},\begin{pmatrix}
3\\
3\\
3\end{pmatrix},\begin{pmatrix}
1\\
2\\
3\end{pmatrix},\begin{pmatrix}
2\\
3\\
1\end{pmatrix},\begin{pmatrix}
3\\
1\\
2\end{pmatrix},\begin{pmatrix}
1\\
3\\
2\end{pmatrix},\begin{pmatrix}
2\\
1\\
3\end{pmatrix},\begin{pmatrix}
3\\
2\\
1\end{pmatrix}\Bigg\}$ $\mathbf{.}$

	%
	%\[ \left( \begin{array}{ccccccccc}
%
      %2 & 3 & 1 & 2 & 3 & 1 & 2 & 3 & 1 \\
			%2 & 3 & 1 & 3 & 1 & 2 & 1 & 2 & 3 \\
			%2 & 3 & 1 & 1 & 1 & 3 & 3 & 1 & 2	\end{array} \right)\]

{\bf Step-2:} Now removing the last entries from each $3-$tuple of $S'$, we get $S''$ that contains nine $2-$tuples as

$S''=\Bigg\{\begin{pmatrix}
1\\
1\end{pmatrix},\begin{pmatrix}
2\\
2\end{pmatrix},\begin{pmatrix}
3\\
3\end{pmatrix},\begin{pmatrix}
1\\
2\end{pmatrix},\begin{pmatrix}
2\\
3\end{pmatrix},\begin{pmatrix}
3\\
1\end{pmatrix},\begin{pmatrix}
1\\
3\end{pmatrix},\begin{pmatrix}
2\\
1\end{pmatrix},\begin{pmatrix}
3\\
2\end{pmatrix}\Bigg\}$ $\mathbf{.}$

	%\[ \left( \begin{array}{ccccccccc}
%
      %2 & 3 & 1 & 2 & 3 & 1 & 2 & 3 & 1 \\
			%2 & 3 & 1 & 3 & 1 & 2 & 1 & 2 & 3  \end{array} \right)\]

{\bf Step-3:} Let us add $(\mathbf{-1})$ to each of the entries of $S$ and then multiply each of the tuples by $3$. Now to each $2-$tuple add all $2-$tuples of $S''$ and generate $S'''$ with $|S'''|=36$. Now all the entrices of $S'''$ are less than or equal to $6$. All the tuples of $S'''$ are as follows:

\tiny{$\Bigg\{\begin{pmatrix}
1\\
1\end{pmatrix}$, $\begin{pmatrix}
4\\
4\end{pmatrix}$, $\begin{pmatrix}
1\\
4\end{pmatrix}$, $\begin{pmatrix}
4\\
1\end{pmatrix}$, $\begin{pmatrix}
2\\
2\end{pmatrix}$, $\begin{pmatrix}
5\\
5\end{pmatrix}$, $\begin{pmatrix}
2\\
5\end{pmatrix}$, $\begin{pmatrix}
5\\
2\end{pmatrix}$, $\begin{pmatrix}
3\\
3\end{pmatrix}$, $\begin{pmatrix}
6\\
6\end{pmatrix}$, $\begin{pmatrix}
3\\
6\end{pmatrix}$, $\begin{pmatrix}
6\\
3\end{pmatrix}$, $\begin{pmatrix}
1\\
2\end{pmatrix}$, $\begin{pmatrix}
4\\
5\end{pmatrix}$, $\begin{pmatrix}
1\\
5\end{pmatrix}$, $\begin{pmatrix}
4\\
2\end{pmatrix}$, $\begin{pmatrix}
2\\
3\end{pmatrix}$, $\begin{pmatrix}
5\\
6\end{pmatrix}$, $\begin{pmatrix}
2\\
6\end{pmatrix}$, $\begin{pmatrix}
5\\
3\end{pmatrix}$, $\begin{pmatrix}
3\\
1\end{pmatrix}$, $\begin{pmatrix}
6\\
4\end{pmatrix}$, $\begin{pmatrix}
3\\
4\end{pmatrix}$, $\begin{pmatrix}
6\\
1\end{pmatrix}$, $\begin{pmatrix}
1\\
3\end{pmatrix}$, $\begin{pmatrix}
4\\
6\end{pmatrix}$, $\begin{pmatrix}
1\\
6\end{pmatrix}$, $\begin{pmatrix}
4\\
3\end{pmatrix}$, $\begin{pmatrix}
2\\
1\end{pmatrix}$, $\begin{pmatrix}
5\\
4\end{pmatrix}$, $\begin{pmatrix}
2\\
4\end{pmatrix}$, $\begin{pmatrix}
5\\
1\end{pmatrix}$, $\begin{pmatrix}
3\\
2\end{pmatrix}$, $\begin{pmatrix}
6\\
5\end{pmatrix}$, $\begin{pmatrix}
3\\
5\end{pmatrix}$, $\begin{pmatrix}
6\\
2\end{pmatrix}\Bigg\}$} $\mathbf{.}$

%\tiny{
 %\[ \left( \begin{array}{*{36}c}
%%\begin{pmatrix} 
%5 & 6 & 4 & 5 & 6 & 4 & 5 & 6 & 4 & 2 & 3 & 1 & 2 & 3 & 1 & 2 & 3 & 1 & 5 & 6 & 4 & 5 & 6 & 4 & 5 & 6 & 4 & 2 & 3 & 1 & 2 & 3 & 1 & 2 & 3 & 1 \\
%5 & 6 & 4 & 6 & 4 & 5 & 4 & 5 & 6 & 2 & 3 & 1 & 3 & 1 & 2 & 1 & 2 & 3 & 2 & 3 & 1 & 3 & 1 & 2 & 1 & 2 & 3 & 5 & 6 & 4 & 6 & 4 & 5 & 4 & 5 & 6 
%%\end{pmatrix}
%\end{array} \right)\]}
\normalsize
 			
{\bf Step-4:} From each $2-$tuple of $S'''$, we form a column consisting of two blocks of size $6$ where the positions of $1's$ in each block are dictated by the two entries of the tuple. Finally, from all these steps, we obtain the following matrix 
$\Phi_{12 \times 36} :$ 		

\tiny{\[ \left( \begin{array}{*{36}c}
1&0&1&0&0&0&0&0&0&0&0&0&1&0&1&0&0&0&0&0&0&0&0&0&1&0&1&0&0&0&0&0&0&0&0&0\\
0&0&0&0&1&0&1&0&0&0&0&0&0&0&0&0&1&0&1&0&0&0&0&0&0&0&0&0&1&0&1&0&0&0&0&0\\
0&0&0&0&0&0&0&0&1&0&1&0&0&0&0&0&0&0&0&0&1&0&1&0&0&0&0&0&0&0&0&0&1&0&1&0\\
0&1&0&1&0&0&0&0&0&0&0&0&1&0&1&0&0&0&0&0&0&0&0&0&0&1&0&1&0&0&0&0&0&0&0&0\\
0&0&0&0&0&1&0&1&0&0&0&0&0&0&0&0&0&1&0&1&0&0&0&0&0&0&0&0&0&1&0&1&0&0&0&0\\
0&0&0&0&0&0&0&0&0&1&0&1&0&0&0&0&0&0&0&0&0&1&0&1&0&0&0&0&0&0&0&0&0&1&0&1\\
\hline
1&0&0&1&0&0&0&0&0&0&0&0&0&0&0&0&0&0&0&0&1&0&0&1&0&0&0&0&1&0&0&1&0&0&0&0\\
0&0&0&0&1&0&0&1&0&0&0&0&1&0&0&1&0&0&0&0&0&0&0&0&0&0&0&0&0&0&0&0&1&0&0&1\\
0&0&0&0&0&0&0&0&1&0&0&1&0&0&0&0&1&0&0&1&0&0&0&0&1&0&0&1&0&0&0&0&0&0&0&0\\
0&1&1&0&0&0&0&0&0&0&0&0&0&0&0&0&0&0&0&0&0&1&1&0&0&0&0&0&0&1&1&0&0&0&0&0\\
0&0&0&0&0&1&1&0&0&0&0&0&0&1&1&0&0&0&0&0&0&0&0&0&0&0&0&0&0&0&0&0&0&1&1&0\\
0&0&0&0&0&0&0&0&0&1&1&0&0&0&0&0&0&1&1&0&0&0&0&0&0&1&1&0&0&0&0&0&0&0&0&0\\
\end{array} \right)\]}

	\normalsize

\subsection{Example}
Let $p_{1}, _{2}$ be two distinct primes and choose a positive integer $r$, such that $r<p_{2}<p_{1}$. Using the method in \cite{Ronald_2007}, one obtains binary matrices $\Psi, \Psi'$ of sizes $p^{2}_{2} \times p^{r+1}_{2}$ and $p^{2}_{1} \times p^{r+1}_{1}$ respectively. Thus we get $p^{r+1}_{1}$ number of $p_{1}-$tuples and $p^{r+1}_{2}$ number of $p_{2}-$tuples. If we apply our construction procedure on these matrices then we generate matrix $\Phi$ of size $p_{1}p^{2}_{2} \times (p_{1}p_{2})^{r+1}$ with coherence $\frac{r}{p_{2}}$. The density of this matrix $\Phi$ is $\frac{1}{p_{1}p_{2}}$, which is small compared to $\frac{1}{p_{1}}$ and $\frac{1}{p_{2}}$, the densities of $\Psi$ and $\Psi'$ respectively.
\par In the construction of $\Phi$, the row size is $p_{1}p^{2}_{2}$ and each column contains $p_{2}$ number of ones, and the overlap between any two columns is $r$, therefore the maximum possible column size for this construction \cite{amini_2011} is $\frac{{p_{1}p^{2}_{2}\choose r+1}}{{p_{2} \choose r+1}}$, which is of the order $(p_{1}p_{2})^{r+1}$, that is, $\frac{{p_{1}p^{2}_{2}\choose r+1}}{{p_{2} \choose r+1}} = O((p_{1}p_{2})^{r+1})$. Consequently, our enlarged matrix attains the maximum possible column size asymptotically. 
 
%\textbf{Remark-2:} In our previous work \cite{ram_2013}, using Euler Squares, we have constructed binary sensing matrices of any row size $m$ except for $ p, p^2$ (for prime $p$) with asymptotically optimal column size such that overlap between any two columns is $1$. In the present work, however, we construct sparse binary sensing matrices of any row size $m$ except for $p, p^2, pq$ for distinct primes $p,q$ such that overlap between any two columns is $r' \geq 1$ and these matrices too possess optimal possible column size asymptotically. Because of $r'>1$, the matrix so constructed here has better column-to-row 
%ratio than the one in \cite{ram_2013}. This has greater relevance to applications like dimensionality reduction \cite{Elad_2010}.  
%%%%%%%%%%%%%%%%%%%%%%%%%%%%%%%%%%%%%%%%%%%%%%%%%%%%%%%%%%%%%%%%%%%
\section{An application of proposed construction procedure:}
%%%%%%%%%%%%%%%%%%%%%%%%%%%%%%%%%%%%%%%%%%%%%%%%%%%%%%%%%%%%%%%%%%%%
\ As stated already, the present work helps us to construct deterministic CS matrices of lower density. In addition, we show that this methodology could be adopted to constructing CS matrices of general row sizes.

%We can construct binary matrices of row size different from $p, p^2, pq$ (for distinct primes $p,q$), by our construction procedure with the help of existing binary matrices.

\begin{theorem}
Suppose $m$ is any positive integer different from $p, p^2, pq$ for distinct primes $p,q$. Then there exists a binary CS matrix of row size $m$.
 
%If $m=p^{r_{1}}_{1}p^{r_{2}}_{2}\ldots p^{r_{l}}_{l}$ such that $r_{i}\geq 2$ for $1 \leq i \leq l$, $r_{i} > 2$ for $i=1$ then there exist a matrix of row size $m$.
\end{theorem} 
\begin{proof}
%Without loss of generality assume that $2^{r} \leq p^{r_{1}}_{1} \leq p^{r_{2}}_{2} \leq \ldots \leq p^{r_{l}}_{l}$.\\
\textbf{Case 1:} 
If $m = p^{i}, i>2$, then $m$ can be written as $m = p^{i-2}p\cdot p$  (let take $p=k\leq n'=p\leq p^{i-2}=n$). Suppose $\Psi$ is a matrix having $p'(>p)-$blocks where each block is of size $p^{i-2}$ with intersection between any two different columns is $r$. Let $\Psi'$ be a CS matrix having $p$ blocks where each block is of size $p$  with intersection between any two different columns is $r'$(It is guaranteed in \cite{Ronald_2007} that such matrices exist and in above Example 3.2). By applying our construction procedure on these two matrices, we generate matrix $\Phi$ of row size $p^{i}$.\\ %(if we work with the matrices constructed in \cite{Ronald_2007} and $r''=\max\{r,r'\}$, then the coherence of the resulting matrix $\Phi$ constructed by this choice of construction is $\frac{r''}{p}$ and column size is $p^{r(i-2)+r'+2}$)\\ 
\textbf{Case 2:} Suppose $m=p^{r_{1}}_{1}p^{r_{2}}_{2}$ such that both $r_{1} \neq 1 \neq r_{2}$ and $p^{r_{1}}_{1}>p^{r_{2}}_{2}$. Take a CS matrix $\Psi$ having $p_{1}-$ blocks with each block being of size $p_{1}^{r_{1}-1}$. Let $\Psi'$ be another CS matrix having $p_{1}-$blocks where each block is of size $p_{2}^{r_{2}}$ (existence of such matrices is guaranteed in \cite{Ronald_2007}). By applying our construction procedure on these two matrices, we generate matrix $\Phi$ of row size $m$.\\ 
\textbf{Case 3:} Suppose $m=p^{r_{1}}_{1}p^{r_{2}}_{2}\cdots p^{r_{s}}_{s}$ such that $ s\geq3$ and $ p^{r_{1}}_{1}>p^{r_{2}}_{2}> \cdots >p^{r_{s}}_{s}$. By induction on $s$, we show that there exists a CS matrix of row size $m$. For $s=3$, set $m_3=p^{r_{1}}_{1}p^{r_{2}}_{2}p^{r_{3}}_{3}$. Suppose the CS matrix $\Psi$ has $p^{r_{3}}_{3}-$blocks where each block is of size $p_{1}^{r_{1}}$. Let $\Psi'$ be another CS matrix having $p^{r_{3}}_{3}-$blocks such that each block is of size $p_{2}^{r_{2}}$ (existence of such matrices is guaranteed in \cite{Ronald_2007}). If we apply our construction procedure on these two matrices, then we generate matrix $\Phi$ of row size $p^{r_{1}}_{1}p^{r_{2}}_{2}p^{r_{3}}_{3}$, which is $m_3$. Assume that for $3<l<k$, there exists a matrix of row size $m_l=p^{r_{1}}_{1}p^{r_{2}}_{2}\cdots p^{r_{l}}_{l}$. Now we have to show that there exists a CS matrix of row size $m_k=p^{r_{1}}_{1}p^{r_{2}}_{2}\cdots p^{r_{k-2}}_{k-2}p^{r_{k-1}}_{k-1}p^{r_{k}}_{k}$. Take a CS matrix $\Psi$ having $p^{r_{k}}_{k}-$blocks such that each block is of size $m_{k-2}=p^{r_{1}}_{1}p^{r_{2}}_{2}\cdots p^{r_{k-2}}_{k-2}$. Suppose $\Psi'$ has $p^{r_{k}}_{k}-$blocks such that each block is of size $p_{k-1}^{r_{k-1}}$. By applying our construction procedure on these two matrices, we generate matrix $\Phi$ whose row size is $p^{r_{1}}_{1}p^{r_{2}}_{2}\cdots p^{r_{k-2}}_{k-2}p^{r_{k-1}}_{k-1}p^{r_{k}}_{k}$, which is $m$. By induction we conclude that there exists a CS matrix of row size $m=p^{r_{1}}_{1}p^{r_{2}}_{2}\cdots p^{r_{s}}_{s}.$
\end{proof}

\noindent \textbf{Remark-3:} In the process of obtaining a binary matrix of row size $m$ if we use matrices $\Psi_{1},\ldots,\Psi_{t}$ with column sizes $M_{1},\ldots,M_{t}$ respectively and with coherence $\mu_{1},\ldots, \mu_{t}$ respectively, then the resulting matrix $\Phi$ is of size $m \times (M_{1}\cdots M_{t})$ and the coherence $\mu_{\Phi}$ of $\Phi$ is $\text{max}{\{\mu_{1},\ldots, \mu_{t}\}}.$

\noindent \textbf{Remark-4:}
The authors of \cite{amni_2012} obtained general row size via Kronecker Delta product. If we apply Kronecker Delta product on $\Psi_{m \times M}$ and $\Psi'_{m' \times M'}$, the existing matrices that we started with, one obtains $\Psi''$, which is of the size $mm' \times MM'$. The method proposed herein, on the other-hand, results in $\Phi$ of size $nn'k \times MM'$. Although both the matrices have same coherence and same density, as $nn'k<mm'$, $\Phi$ has better aspect ratio or redundancy factor (column to row ratio) compared to $\Psi'$.\

\noindent \textbf{Remark-5:}
 As coherence of the resulting matrix depends on the choice of the two matrices, we can choose the two matrices in such a way that the resultant matrix has low coherence. In the above proof, we have used as an example the matrices constructed in \cite{Ronald_2007}. Nevertheless, one may use other suitable binary matrices like those constructed in \cite{li_2012}. \

\section{Construction of ternary matrices}
%%%%%%%%%%%%%%%%%%%%%%%%%%%%%%%%%%%%%%%%%%%%%%%%%%%%%%%
Let $\Phi_{m\times M}$ be a matrix having $k$ number of blocks where each block is of size $n$, containing single one's and intersection between any two different columns of $\Phi$ is atmost $r$. Let the $i^{th}$ column of $\Phi$ be  $\Phi_{i}$. For $1 \leq i \leq M$, define $f_{i} =  (((\text{supp}(\Phi_{i})-1)(\mod n))^{\bf{T}}+\bf{1}^{T})$, 
where $\bf{1}^{T}$ is the vector consisting of all ones of size $ k \times 1$, supp$(\Phi_{i})$ is the support vector of $\Phi_{i}$. Since $\Phi$ has $k$ blocks and each block contains one $1$, it follows that $|f_{i}| = k$ where $f_{i}$ is a $k-$ tuple on the set $X=\{1,2,\ldots,n\}$. Since $\Phi$ has $M$ columns, we have $M$ such $k-$tuples. Let $f_{i}=(f_{i_1}, f_{i_2}, \ldots,f_{i_k})$ be the $k-$ tuples. Now decompose this $k-$tuple into $k-$number of two tuples as $(l,f_{i_{l}})$, where $l$ denotes the $l^{th}$ block position and $f_{i_{l}}$ denotes the $l^{th}$ entry of $f_{i}$. It is to be noted here that $1 \leq l \leq k, 1 \leq f_{i_{l}} \leq n$.
%%%%%%%%%%%%%%%%%%%%%%%%%%%%%%%%%%%%%%%%%%%%%%%%%%%%
\subsection{Construction}
%%%%%%%%%%%%%%%%%%%%%%%%%%%%%%%%%%%%%%%%%%%%%%%%%
 From the $k-$ tuple $f_{i}=(f_{i_1}, f_{i_2}, \ldots,f_{i_k}),$ we form a binary vector of length $n \times k$ where $1$ occurrs in the positions $(l-1).n+f_{i_l}$ for $l= 1,2,....,k$ and rest of the positions are zeros (the vector is nothing but $\Phi_{i}$ )  and then we replace $1$ with $(-1)$ if $l>f_{i_{l}}$, which results in a vector consisting of $0,1,-1$ as entries. As we have $M$ such $k-$tuples, we get $M$ number of ternary vectors of length $n \times k$. This way, we get a ternary matrix $\Phi'_{m\times M}$ of size same as that of $\Phi_{m\times M}$. 
%%%%%%%%%%%%%%%%%%%%%%%%%%%%%%%%%%%%%%%%%%%%%%%%%%%%%%%%%%%%
\par Let $(0\; 1\; 0\; 0\; 0\; 1\; 1\; 0\; 0)^{T}$ be a column from a matrix having $3$ blocks and each block is of size $3$. Then the support of this column is the $3-$ tuple $(2\; 3\; 1)$. Break this $3$- tuple into three $2$- tuples $(1\;2), (2\;3), (3\;1)$ as we mentioned above.  Now in the 3rd block as $1$ occurs in the 1st position, we replace $1$ with $-1$ in the 3rd block of the binary vector  $(0\; 1\; 0\; 0\; 0\; 1\; 1\; 0\; 0)^{T}$ to generate the ternary vector $(0\; 1\; 0\; 0\; 0\; 1\; (-1)\; 0\; 0)^{T}$.  

\noindent \textbf{Remark-6:} The coherence and density of the ternary matrix $\Phi'$ remain same as that of the binary matrix $\Phi$. \\
\par As in the case of binary CS matrices, it is possible to construct ternary matrices of general row size as well, which is concluded as the following theorem:
\begin{theorem}
Suppose $m$ is any positive integer different from $p, p^2, pq$ for distinct primes $p,q$. Then there exists a ternary CS matrix of row size $m$. 
\end{theorem}
\begin{proof} 
Initially we construct a binary matrix of row size $m$, which is discussed in Section 4.  Then we use methodology in section 5.1  to obtain a ternary matrix of row size $m$.
\end{proof}
\subsection{A different approach to constructing ternary matrix:}
In this section, we present another deterministic construction procedure of ternary matrix by combining binary and Hadamard matrices . We show that the resulting matrix has same density and coherence as binary but with better aspect ratio (column to row ratio). Let $\Psi_{m \times M}$ be a binary sensing  matrix having $k$ number of one's in each column and overlap between any two columns is $r$. Suppose there exists a Hadamard matrix $H$ of size $(k+r')$ for some $r' \in { \{0,\dots, r\}}$. For each column of $\Psi$, we replace each of its $1$-valued entries with a distinct row of $H$ to obtain a ternary CS matrix $\Phi$ of size $m \times M(k+r')$. As the rows of Hadamard matrix $H$ are orthogonal, the rows of $\Phi$ are orthogonal. From the construction methodology, it is easy to check that the coherence of the matrix $\Phi$ is $\frac{r}{k}$. The density of the matrix $\Phi$ is $\frac{k}{m}$.

\section{Concluding Remarks:}
As CS matrices of low density (or high sparsity) result in algorithms with low computational complexity, the present work has constructed sparse CS matrices from the existing two binary sensing matrices. As an application of proposed methodology, we have shown that binary and ternary CS matrices for a more general set of numbers can be constructed.

\section{\bf Acknowledgments}
The first author is thankful for the support (Ref No. 19-06/2011(i)EU-IV) that he receives from UGC, Govt of India. The second author gratefully acknowledges the support (Ref No. 20-6/2009(i)EU-IV) that he receives from UGC, Govt of India. The third author is thankful to DST (SR/FTP/ETA-054/2009) for the partial support that he received. We thank Mr. Roopak R Tamboli for helping us in simulation work.

\end{document}